\documentclass{article}

\usepackage{arxiv}

\usepackage[utf8]{inputenc} 
\usepackage[T1]{fontenc}    
\usepackage{hyperref}       
\usepackage{url}            
\usepackage{booktabs}       
\usepackage{amsfonts}       
\usepackage{nicefrac}       
\usepackage{microtype}      
\usepackage{lipsum}
\usepackage{graphicx}
\graphicspath{ {./images/} }

\usepackage[utf8]{inputenc}

\usepackage{amsmath}
\usepackage{amssymb}
\usepackage{amsthm}
\usepackage{amsfonts}
\usepackage{graphicx}
\usepackage{multirow}
\usepackage{xcolor}
\usepackage{hyperref}

\usepackage{subcaption}

\newtheorem{theorem}{Theorem}
\newtheorem{lemma}{Lemma}

\newtheorem{idconst}{Identifiability constraint}

\usepackage{natbib}

\title{Model-based sparse mixed-type PCA}

\author{
 Lauri Heinonen \\
  Department of Mathematics and Statistics\\
  University of Turku\\
  Turku, Finland \\
  \texttt{lauri.k.heinonen@utu.fi} \\
   \And
 Joni Virta \\
  Department of Mathematics and Statistics\\
  University of Turku\\
  Turku, Finland \\
  \texttt{joni.virta@utu.fi}
}

\date{June 2026}

\begin{document}
\maketitle
\begin{abstract}
This work presents a new method for principal component analysis (PCA) of a mixed-type data consisting of continuous, binary, integer-valued and positive continuous variables. The data are assumed to come from a probability model, where the parameters of the exponential family distributions are determined by a set of shared Gaussian latent variables. The proposed method, MTPCA, is based on estimating the covariance matrix of these latent mixtures through the method of moments. A way to sparsify the component loadings is presented and aligns with the classical theory of sparse PCA. We propose a strategy for estimating the principal component scores and discuss the choice of the latent dimension. The method's performance is studied with a simulated mixed-type data and we illustrate the model on the Zoo data set consisting of binary animal characteristics.
\end{abstract}


The work of LH and JV was supported by the Research Council of Finland (grants 347501, 353769, 368494).

\section{Introduction}

\subsection{Problem description and past approaches}

A common question in statistics is whether there exists some latent structure behind the observed variables. The role of this question does not decrease in the age of increasing amounts of data. The most common way to find these factors is principal component analysis (PCA) \citep{jolliffe2002principal}. The output of PCA consists of two key components: the \textit{loadings} that give the strength of connection between the observed variables and the latent factors, and the \textit{scores} that give the estimated values of the factor variables for each observation. One key limitation of PCA is that it is best suited for continuous-type variables having real values. More precisely, regular PCA arises naturally from assuming a Gaussian distribution for the observed variables as in probabilistic PCA \citep{tipping1999probabilistic}. However, it is clear that similar types of questions about latent factors, loadings and scores arise also for variables of binary, integer-valued or positive continuous type that are regularly observed in almost all applications.

Different kinds of approaches have been taken to develop methods analogous to PCA for various non-Gaussian variables or datasets of mixed-type. Some are specific to a variable type or combination thereof: \cite{kolenikov2004use} developed PCA based on polychoric correlations for discrete/binary data, \cite{virta2023poisson} used the multivariate Poisson-log normal distribution \citep{aitchison1989multivariate} to formulate a tractable PCA-model for data consisting of count-valued matrices, and \cite{chavent2022multivariate} used generalized SVD to carry out a hybrid of PCA and multiple correspondence analysis for data consisting of both numeric and categorical data.

A particularly popular approach to mixed-type PCA is to use properties of the (natural) exponential family, employing approaches similar to generalized linear models: In \cite{collins2001generalization}, the natural parameter of an exponential family distribution is presented as a linear combination of principal components through a link function. The components are estimated through minimizing the Bregman divergence, corresponding to maximizing the log-likelihood of the data distribution. The solution is found one component at a time by an algorithm alternating between the score and loading vectors. Another GLM-based method is the generalized factor model \citep{Liu03042023} in which the latent factors are modeled as non-random and the authors give a likelihood-based rule to determine a suitable number of factors. They propose a penalized extension that can be used to promote, for example, sparsity. Also the behavior of the estimators, when $n\to\infty$ and $p\to\infty$, is studied. In \cite{kong2024generalized} a generalized matrix factor model is presented for matrix-valued data consisting of mixed-type entries and the authors derive consistent estimators for the model parameters under a high-dimensional setting. Generalized latent trait models \citep{moustaki2000generalized} are a similar type of generalization coming from the background of latent trait models in psychology.

Similar kinds of approaches have also been taken for other close or related problems: In \cite{10.1093/biomtc/ujaf065} an algorithm is developed for sufficient dimension reduction with discrete predictors. In \cite{LIU2023107822} and \cite{liu2025high}, the generalized factor model from \cite{Liu03042023} is used for imputation of missing values.

\subsection{Our contributions}

The previous methods for mixed-type PCA are mostly algorithmic, focusing on how to solve the latent components using optimization. In this work, we develop a statistical approach, MTPCA (for Mixed-Type PCA), that is based on a well-defined and interpretable probability model. As is common to statistical factor models, we employ two levels of random variables, which are connected by the relation
\begin{align}\label{eq:model_simplification}
    X_j \mid Z \sim \mathcal{D}_j( g_j(u_j'Z) ), \quad Z \sim N_d(0, \Lambda),
\end{align}
where $\mathcal{D}_j$ is an exponential family distribution, $g_j$ a link function and $u_j$ a loading vector, $j = 1, \ldots, p$. Given the observed $X_1, \ldots, X_p$ (or a sample thereof), the objective in model \eqref{eq:model_simplification} is to estimate the loading vectors $u_1, \ldots, u_p$ and the latent score vector $Z$. The model \eqref{eq:model_simplification} can essentially be seen as an extension of probabilistic PCA (PPCA) \citep{tipping1999probabilistic} to mixed-type data.


To estimate the loadings $u_j$, we show how the covariances of the latent combinations $u_j'Z$ can be computed using the method of moments, after which regular PCA can be carried out on this covariance matrix to find the principal component loadings. The scores $Z$ can then be estimated by maximizing a specific, strictly concave log-likelihood function.

Compared to the earlier works, our model admits explicit closed-form formulas for the covariances of the combinations of the principal components in almost all cases. Only in the case of binary-to-binary interaction is the corresponding covariance solved through an implicit formula. This enables exact and fast computation without iterative procedures, and instantly guarantees the consistency of the latent covariances estimators. 

Basing the estimation of the loadings on the covariance matrix of the latent combinations $u_j' Z$ also gives us two distinct advantages. Firstly, this lets us naturally accommodate sparsity in a way similar to sparse PCA \citep{zou2006sparse}, allowing us to produce sparse estimates of loading vectors $u_1, \ldots, u_p$. Of the previous mixed-type PCA models, sparsity has been discussed only in \cite{Liu03042023}, and even there very briefly. Secondly, this lets the method simultaneously accommodate non-linear estimation (via the link functions) but also benefit from all the standard covariance matrix based tools familiar to users of PCA, i.e., scree plot, Kaiser rule, parallel plots, factor rotations, etc. Of these, we demonstrate sparsity (Section \ref{sec:sparsity}) and scree plots (Section \ref{sec:examples}) later in the work.


\subsection{Manuscript structure}

Section \ref{sec:proposal} formulates the proposed MTPCA method in several parts: In Sections \ref{sec:model} and \ref{sec:estimation}, we carefully describe our model and give the formulas for estimating the model parameters. We also show how to impose sparsity on the model using sparse PCA (Section \ref{sec:sparsity}), how to estimate the principal component scores (Section \ref{sec:scores}) and how to decide on the correct number of components (Section \ref{sec:number_of_components}). In Section \ref{sec:examples}, we study the method's performance using simulated mixed-type data and illustrate it further in the context of the Zoo data consisting of binary variables.

\section{The MTPCA procedure} \label{sec:proposal}

\subsection{Model} \label{sec:model}



We begin by giving a more detailed description of our proposed model \eqref{eq:model_simplification}.  Let $Z \sim \mathrm{N}_d(0, \Lambda)$ represent the latent vector of multivariate normal random variables with diagonal covariance matrix $\Lambda = \mathrm{diag}(\lambda_1, \ldots, \lambda_d)$ and the variances $\lambda_1 > \lambda_2 > \cdots > \lambda_d > 0 $. Assume that our observed random vector $X$ of length $p$ is partitioned into four subvectors $X_1 \in \mathbb{R}^{p_1}$, $X_2 \in \mathbb{R}^{p_2}$, $X_3 \in \mathbb{R}^{p_3}$, $X_4 \in \mathbb{R}^{p_4}$ that take values in $\mathbb{R}^{p_1}$, $\mathbb{R}_+^{p_2}$, $\mathbb{N}_0^{p_3}$, $\{0 ,1\}^{p_4}$, respectively. That is, the vectors contain real-valued (type 1), positive (type 2), count (type~3) and binary (type 4) data, respectively.

Let $W := \mu + U Z$ for a vector $\mu \in \mathbb{R}^p$ and matrix $U \in \mathbb{R}^{p \times d}$ with orthonormal columns and partition $W$ into $W_1, W_2, W_3, W_4$ similarly to $X$. Assume further that the observed data $X$ depends on the vector $W$ as follows:
\begin{align}\label{eq:our_model}
\begin{split}
    X_{1j} | Z &\sim \mathrm{N}(W_{1j}, \tau_j^2), \quad j = 1, \ldots, p_1, \\
    X_{2j} | Z &\sim \mathrm{Exp}(e^{W_{2j}}), \quad j = 1, \ldots, p_2, \\
    X_{3j} | Z &\sim \mathrm{Poi}(e^{W_{3j}}), \quad j = 1, \ldots, p_3, \\
    X_{4j} | Z &\sim \mathrm{Ber}(\Phi(W_{4j})), \quad j = 1, \ldots, p_4,
\end{split}
\end{align}
where $\Phi$ is the CDF of the standard normal distribution, and each variable type has its own distribution and corresponding link function. We assume that all elements of $X$ are conditionally independent conditional on $Z$, meaning that \eqref{eq:our_model} fully specifies the distribution of the vector $X$. For type 1 (conditionally Gaussian continuous variables), we further assume that all variance parameters $\tau_j^2$ satisfy $\tau_j^2 > 0$. To summarize, $X$ consists of $p$ mixed-type observed variables that are all driven by the shared set of $d$ normal latent variables collected in $Z$. The parameters of interest are $\mu = \mathrm{E}(W) \in \mathbb{R}^p$ and the latent covariance matrix $\Sigma := \mathrm{Cov}(W) = U \Lambda U'$, based on whose truncated eigen decomposition the loadings $U \in \mathbb{R}^{p \times d}$ and the principal component variances on the diagonal of $\Lambda \in \mathbb{R}^{d \times d}$ can be determined.

As stated earlier, \eqref{eq:our_model} is a natural extension of probabilistic principal component analysis \citep{tipping1999probabilistic} for mixed-type data. However, contrary to PPCA, in model \eqref{eq:our_model} observation noise is not modelled with additive error terms (except in the case of type~1 variables), but rather comes from the conditional sampling distributions, analogously to generalized linear models.

Due to the presence of multiple variables types, we use in the sequel indexing conventions that differ somewhat from the standard ones. Namely, we let $\mu_i$ denote an element of the mean vector $\mu = \mathrm{E}(W)$ and $\sigma_i^2$ a diagonal element of the matrix $\Sigma = \mathrm{Cov}(W)$ corresponding to a \textit{generic variable of type $i$}. That is, $\mu_2$ does not represent the second element of $\mu$, but the mean value of $W$ for a generic variables of type 2. This convention allows presenting the estimation procedure below in a succinct manner. Similarly, let $\sigma_{ij}$ be a non-diagonal element of $\Sigma$ corresponding to generic variables of types $i$ and $j$, with $\sigma_{ii'}$ denoting the element corresponding to two different generic variables of type $i$.

The model parameters for the 0/1-Bernoulli variables (type 4) are not fully identifiable since $\mathrm{E}(X_{4j}) = \mathrm{E}(X_{4j}^2)$. Hence, we impose a constraint that connects the corresponding mean and variance parameters, and as such allows identifying them.

\begin{idconst}\label{id:bernoulli}
    Each of the Bernoulli-variables (type 4) satisfies
    \begin{align*}
        \frac{\mu_4}{\sqrt{1+\sigma_4^2}} = s \sqrt{ \sigma_4^2 },
    \end{align*}
    for some sign $s \in \{ -1 , 1 \}$.
\end{idconst}

Identifiability constraint \ref{id:bernoulli} does not limit the range of type 4 marginal distributions that fall under the model \eqref{eq:our_model}.

As in probabilistic PCA \citep{tipping1999probabilistic}, we assume that the normal variances are all equal, $\tau_1^2 = \cdots = \tau_{p_1}^2$, making the Gaussian part have isotropic errors. For now, we also assume that the $\tau_j^2$ are known and will discuss how these can be estimated in the next subsection. 






\subsection{Estimation of the model parameters} \label{sec:estimation}

Assume that we have observed a sample $X_1, \ldots, X_n$ from the proposed model \eqref{eq:our_model}. MTPCA estimates the parameters $\mu, \Sigma$ using the method of moments. This requires computing the first moments and the second moments and cross-moments for the four different types of variables and solving for the unknown parameters. In all but one case this leads to closed-form estimators, given below in Theorem \ref{theo:estimates}. The sole exception is the second cross-moment between two Bernoulli-variates, which we discuss separately below and which we solve numerically in practice. Theorem \ref{theo:estimates} follows directly from the auxiliary Theorem \ref{theo:moments} in Appendix \ref{sec:proofs}, along with the law of large numbers and the continuous mapping theorem. Below, we denote by $\overline{X_i}$ the sample mean of a generic variable of type $i$, by $\overline{X_i^2}$ the corresponding sample second moment, and by $\overline{X_i X_j}$ the sample second cross-moment between generic variables of types $i$ and $j$.


\begin{theorem}\label{theo:estimates}
Assume that Identifiability constraint \ref{id:bernoulli} holds. Then the following are consistent estimators of the corresponding model parameters.
\begin{align*}
    \hat{\mu}_1 &:= \overline{X_1} \\
    \hat{\mu}_2 &:= -\frac{1}{2}\log2-2\log\overline{X_2}+\frac{1}{2}\log\overline{X_2^2} \\
    \hat{\mu}_3 &:= 2\log\overline{X_3}-\frac{1}{2}\log\left(\overline{X_3^2}-\overline{X_3}\right)\\
    \hat{\mu}_4 &:= \Phi^{-1}(\overline{X_4})\sqrt{ 1 + \left\{\Phi^{-1}(\overline{X_4})\right\}^2 }\\
    \hat{\sigma}_1^2 &:= \overline{X_1^2}-\overline{X_1}^2 - \tau_j^2 \\
    \hat{\sigma}_2^2 &:= -\log2+\log\overline{X_2^2}-2\log\overline{X_2}\\
    \hat{\sigma}_3^2 &:= \log\left(\overline{X_3^2}-\overline{X_3}\right)-2\log\overline{X_3} \\
    \hat{\sigma}_4^2 &:= \left\{\Phi^{-1}(\overline{X_4})\right\}^2 \\
    \hat{\sigma}_{11'} &:= \overline{X_1 X_{1'}}-\overline{X_1}\ \overline{X_{1'}} \\
    \hat{\sigma}_{22'} &:= \log\overline{X_2 X_{2'}}-\log\overline{X_2}-\log\overline{X_{2'}} \\
    \hat{\sigma}_{33'} &:= \log\overline{X_3 X_{3'}}-\log\overline{X_3}-\log\overline{X_{3'}} \\
    \hat{\sigma}_{12} &:= \frac{\overline{X_1}\ \overline{X_2} - \overline{X_1 X_2}}{\overline{X_2}} \\
    \hat{\sigma}_{13} &:= \frac{\overline{X_1 X_3} - \overline{X_1}\ \overline{X_3}}{\overline{X_3}} \\
    \hat{\sigma}_{14} &:= (\overline{X_1 X_4} - \overline{X_1}\ \overline{X_4})\frac{\sqrt{ 1 + \left\{\Phi^{-1}(\overline{X_4})\right\}^2 }}{\varphi\left(\Phi^{-1}(\overline{X_4})\right)}\\
    \hat{\sigma}_{23} &:= \log\overline{X_2}+\log\overline{X_3}-\log\overline{X_2 X_3} \\
    \hat{\sigma}_{24} &:= \sqrt{ 1 + \left\{\Phi^{-1}(\overline{X_4})\right\}^2 }\left\{\Phi^{-1}(\overline{X_4}) - \Phi^{-1}\left(\frac{\overline{X_2 X_4}}{\overline{X_2}}\right)\right\}\\
    \hat{\sigma}_{34} &:= \sqrt{ 1 + \left\{\Phi^{-1}(\overline{X_4})\right\}^2 }\left\{\Phi^{-1}\left(\frac{\overline{X_3 X_4}}{\overline{X_3}}\right) - \Phi^{-1}(\overline{X_4})\right\}
\end{align*}

\end{theorem}

The only parameter whose estimator is missing from Theorem \ref{theo:estimates} is the Bernoulli cross-moment $\sigma_{44}$. By Theorem \ref{theo:moments_2} in Appendix~\ref{sec:proofs}, a natural estimator for it is $\hat{\sigma}_{44}$ defined implicitly as
\begin{align}
    & \Phi_2\left( \Phi^{-1}(\overline{X_4}) , \Phi^{-1}(\overline{X_{4'}}); \frac{\hat{\sigma}_{44'}}{\sqrt{1 + \{ \Phi^{-1}(\overline{X_4}) \}^2} \sqrt{1 + \{ \Phi^{-1}(\overline{X_{4'}}) \}^2 }} \right) -\overline{X_4 X_{4'}} = 0, \label{eq:bernoulli_bernoulli} 
\end{align}
where $\Phi_2(y_1, y_2; \rho)$ is the cumulative distribution function of the bivariate normal distribution with zero means, unit variances and correlation equal to $\rho$. By \cite{plackett1954reduction}, we have $(\partial/\partial \rho)\Phi_2(y_1, y_2; \rho) > 0$ for all $y_1, y_2$, see also formula~(2) in \cite{genz2004numerical}. This implies that any possible solution $\hat{\sigma}_{44'}$ to \eqref{eq:bernoulli_bernoulli} is unique. A natural search interval is $[-\hat{\sigma}_4 \hat{\sigma}_{4'}, \hat{\sigma}_4 \hat{\sigma}_{4'}]$, since the population level covariance matrix is positive-definite. In practice, when the data does not necessarily obey the model, it can happen that a solution to \eqref{eq:bernoulli_bernoulli} does not exist in the search interval, in which case we take as an estimator the value that minimizes the absolute value of the left-hand side of~\eqref{eq:bernoulli_bernoulli}.

By combining the Bernoulli cross-moment estimator with the quantities in Theorem \ref{theo:estimates}, we then obtain the MTPCA-estimators $\hat{\mu} \in \mathbb{R}^p$ and $\hat{\Sigma} \in \mathbb{R}^{p \times p}$. That is, $\hat{\mu}$ is built from four subvectors, each corresponding to estimators for one type of variable, and $\hat{\Sigma}$ consists of $4 \times 4$ blocks, likewise built from the corresponding second moment and cross-moment estimators. To guarantee the positive semi-definiteness of $\hat{\Sigma}$, we threshold any possible negative eigenvalues to zero, and compute its truncated eigendecomposition to form estimates for the loadings and principal component standard deviations, $\hat{U} \in \mathbb{R}^{p \times d}$ and $\hat{\Lambda} \in \mathbb{R}^{d \times d}$. The loadings have the same identifiability indeterminacy as in standard PCA: the signs of the columns of $\hat{U}$ can be chosen freely and, in case $\hat{\Sigma}$ has non-zero eigenvalues of multiplicity greater than one, then there is more freedom in the choice of the corresponding eigenvectors. However, based on our experience, the latter is unlikely to occur with real data. The truncation of the eigendecomposition requires knowing the true number $d$ of latent variables (dimension of $Z$) and we assume it to be known for now, and discuss choosing it later in Section \ref{sec:number_of_components}.

Our earlier assumption that $\tau^2_i = \tau^2_j =: \tau^2$ for all $i,j$ makes the model have isotropic Gaussian error. One particular consequence of this is that if $p_2 = p_3 = p_4 = 0$, i.e., if the data consists solely of continuous variables (type 1), then the proposed method reduces to probabilistic PCA. Under the isotropy assumption the estimation of $\tau$ can be carried out in the following way: Assume that the real number of principal components $d$ is smaller than the number of Gaussian components $p_1$. Then the eigenvalues of the covariance matrix of all $p_1$ variables of type~1 are $(\lambda_1+\tau^2,\dots,\lambda_d+\tau^2,\tau^2,\dots,\tau^2)$. Hence, a natural, consistent estimator for $\tau^2$ is the average of the $p_1-d$ smallest eigenvalues of the sample covariance matrix of the variables of type 1.

\subsection{Imposing sparsity on the parameters} \label{sec:sparsity}

We next describe how the estimated loadings can be made sparse, by which we mean that some elements of the component loading vectors in $\hat{U}$ are forced to zero in the estimation process. The advantages of sparsity are that the loadings are easier to interpret and sparsity in general can be used to avoid overfitting.


We impose sparsity in MTPCA by first estimating the latent covariance matrix $\hat{\Sigma}$ as described in Section \ref{sec:estimation} and then calculating the sparse loadings $\hat{U}$ by applying to $\hat{\Sigma}$ the method SPCA \citep{zou2006sparse}, which formulates PCA as a regression/minimization problem. In this work, we run SPCA using its more general implementation discussed in \cite{li2007sparse, HEINONEN2026105587} which corresponds to solving the optimization problem, 
\begin{align}\label{eq:sparsity_opt}
    \mathrm{argmin} \left\{ \| (I_p - A B') \hat{\Sigma}^{1/2} \|_F^2 + \sum_{j = 1}^d \nu_j \| b_j \|_1 \right\},
\end{align}
over $B = (b_1, \ldots, b_d) \in \mathbb{R}^{p \times d}$ and $A \in \mathbb{R}^{p \times d}$ such that $A' A = I_d$. Given a solution $\hat{A}, \hat{B}$, the sparse loadings are then found as $\hat{U} = \hat{B}$. The resulting level of sparsity is controlled by the tuning parameters $\nu_1, \ldots, \nu_d$ such that larger values impose greater sparsity on the corresponding component. In practice, their values are easier to choose by instead specifying a desired number $r_j$ of non-zero loadings per vector, and then determining $\nu_j$ implicitly based on this. For simplicity, in this work we use a fixed value of $r_j =: r$ for all components. See \cite{li2007sparse, HEINONEN2026105587} for more details on solving the optimization problem \eqref{eq:sparsity_opt} and its theoretical properties.



\subsection{Estimation of the latent components} \label{sec:scores}

Having estimated the sparse loadings and the other parameters, the remaining task is to predict the values of the latent $d$-vectors $Z_i$. Denote the observed values of the four types of variables for a single subject by $x_{1j_1}$, $x_{2j_2}$, $x_{3j_3}$, $x_{4j_4}$, where $j_k = 1, \ldots, p_k$, and let the full $(p_1 + p_2 + p_3 + p_4)$-dimensional vector of covariates be $x$. We take as the prediction of $Z$ the mode of the conditional distribution of $Z \mid X = x$. This strategy is natural in that, when applied in the context of standard PCA, it leads to the usual principal components, see Lemma 2 in \cite{virta2023poisson}. The same approach was also used to predict latent components in \cite{li2010simple, kenney2021poisson}. 

The following theorem expresses this task as a problem in convex optimization, where we have denoted the conditional log-density by $\ell(z \mid x) := \log f_{Z \mid X}(z \mid x)$.

\begin{theorem}\label{theo:z_prediction}
    The conditional log-density satisfies
    \begin{align*}
        \ell(z \mid x) =& C - \frac{1}{2} z' \Lambda^{-1} z -  \sum_{j = 1}^{p_1} \frac{1}{2\tau_j^2}(x_{1j} - \mu_{1j} - u_{1j}'z)^2 \\
        -& \sum_{j = 1}^{p_2}(x_{2j} e^{\mu_{2j} + u_{2j}' z} - \mu_{2j} - u_{2j}'z) \\
        -& \sum_{j = 1}^{p_3}\{e^{\mu_{3j} + u_{3j}' z} - x_{3j}(\mu_{3j} + u_{3j}'z)\} \\
        +& \sum_{j = 1}^{p_4} \{ x_{4j} \log \Phi(\mu_{4j} + u_{4j}' z) + (1 - x_{4j}) \log [1 - \Phi(\mu_{4j} + u_{4j}' z) ] \},
    \end{align*}
    for some constant $C$ not depending on $z$. Moreover, the function $z \mapsto \ell(z \mid x)$ is strictly concave for all $x$ and admits a unique maximizer.
\end{theorem}

Given an observed sample $x_1, \ldots, x_n \in \mathbb{R}^p$, the parameter estimates $\hat{\mu}, \hat{U}, \hat{\Lambda}$ and the latent dimension $d$, we obtain the predictions $\hat{z}_1, \ldots, \hat{z}_n \in \mathbb{R}^d$ as maximizers of the maps $z \mapsto \ell(z \mid x_i)$, $i = 1, \ldots, n$, where the true parameters have been replaced with their estimates. Theorem \ref{theo:z_prediction} implies that these maximizers can be obtained straightforwardly using standard gradient descent.

\subsection{Estimation of the latent dimension} \label{sec:number_of_components}

Since the true dimension $d$ equals the rank of the covariance matrix $\Sigma$, it is natural to estimate it using tools from standard PCA. When using MTPCA for descriptive purposes (as in our real data example in Section \ref{sec:examples}) this can be achieved with the scree plot or any of the standard rules for selecting the number of components in PCA, see Section 6 in \cite{jolliffe2002principal}. Alternatively, one may also choose a small fixed number of components (often $d = 2$ in PCA) to obtain an interpretable and easily visualizable set of loadings and pair-wise scatter plots of the estimated scores.

For a more inferential approach, the ladle estimator of \cite{luo2016combining} could be used. The main idea behind this is to take $M$ bootstrap samples of the data and based on the bootstrap estimates $\hat{\Sigma}^{(1)}, \ldots , \hat{\Sigma}^{(M)}$ pinpoint the dimensionality where the signal subspace changes into the noise subspace, indicated by an increase in eigenvector variability. However, we postpone studying such approaches to future work.

\section{Data examples} \label{sec:examples}

All the R code for the method, MTPCA, and for all examples can be found in the GitHub repository \url{https://github.com/laxuntus/MTPCA}. The implementation allows directly choosing the number of estimated non-zero coefficients in the manner described in Section \ref{sec:sparsity}, making it simple to use in practice.

\subsection{Simulation}

When dimension reduction needs to be applied to mixed-type or categorical data, the classical PCA is still a very popular approach, see, e.g., \cite{prive2020efficient} in the context of genomic data which consist almost solely of 0/1/2-variables. Indeed, despite the connections between PCA and Gaussianity \citep{tipping1999probabilistic}, it often allows uncovering interesting latent structures even for mixed-type data. Moreover, to account for the differing units of measurement, it is advisable to scale each variable in mixed-type data to unit variance.

With the above in mind, we compared the performance of our MTPCA to regular PCA with a simulation study. Datasets were generated from the latent variable model~\eqref{eq:our_model} with five Gaussian, one exponential, one Poisson and one Bernoulli variable. The true loading matrix $U \in \mathbb{R}^{8 \times 2}$ used in the data generation had two columns, each with $q=5$ (first version) or $q=2$ (second version) non-zero coefficients out of the total $p=8$. Datasets had sizes $n=250$ or $n=1500$. For each individual replication of the simulation experiment, the rows of the loading matrix $U$ were further randomly permuted to have the sparsity act on different types of variables. For each combination of $q$ and $n$, 2000 replications were run, see the codes on the GitHub page for more details on the data-generating mechanism.



The loading matrix was estimated with different methods and parameters, and the mean squared Frobenius ($\ell_2$) error of the estimates was calculated. Our proposed method (MTPCA) was used with no sparsity (MTPCA (ns)), with the exact right amount $r=q$ of non-zero coefficients (MTPCA (s)) and with an overestimated amount $r=q+2$ of non-zero coefficients (MTPCA (s+)). As competitors, we used regular PCA with correlation matrix, along with sparse PCA with $r=q$ and $r=q+2$ non-zero coefficients (SPCA and SPCA+). As described in Section \ref{sec:sparsity}, the sparse PCA method \citep{zou2006sparse} was implemented using the SICS-framework of \cite{HEINONEN2026105587} with identity matrix and correlation matrix as the scatter matrices. A random matrix with two orthonormal columns drawn from the corresponding Haar distribution (rand) was used as a baseline. 

Figure \ref{fig:vertailu_q5} shows the violin plots of the mean squared Frobenius norms between the real and estimated loading matrices for the case $q=5$. It can be seen clearly that the non-sparse variant of MTPCA works the best along with sparse MTPCA with $r=q+2=7$, sparse MTPCA with $r=q$ coming in the third place. Regular PCA and sparse PCA behave worse, but still offer a considerable improvement over the random guess, showing that PCA manages to capture some essential piece of the mixed-type factor structure, as predicted above. Interestingly, for some rare iterations MTPCA has high error, but the performance of the PCA variants never goes below 0.70. It can also be seen that increasing the sample size from $n=250$ to $n=1500$ improves the performances of the different MTPCA-variants, giving empirical evidence of their consistency. However, the same is not true for the different variants of PCA.

Figure \ref{fig:vertailu_q2} shows the analogous violin plots of the errors for the more sparse case with $q=2$. Also here, the non-sparse variant of MTPCA works the best along with sparse MTPCA with $r=q+2=4$, the sparse MTPCA with $r=q$ again following. Also now, increasing the sample size improves the performance of different variants of MTPCA but not of PCA. However, contrary to the earlier simulation, here regular and sparse PCA sometimes also achieve very low errors. We also observe that the performances of the non-sparse and $r = q + 2$ variant of MTPCA are in the current scenario almost identical, despite the large difference in the number of non-zero estimated loadings. Thus, the setting demonstrates that the benefits of sparsity in interpretation can be achieved without significant loss in performance.

\begin{figure}
    \centering
    \includegraphics[width=0.9\linewidth]{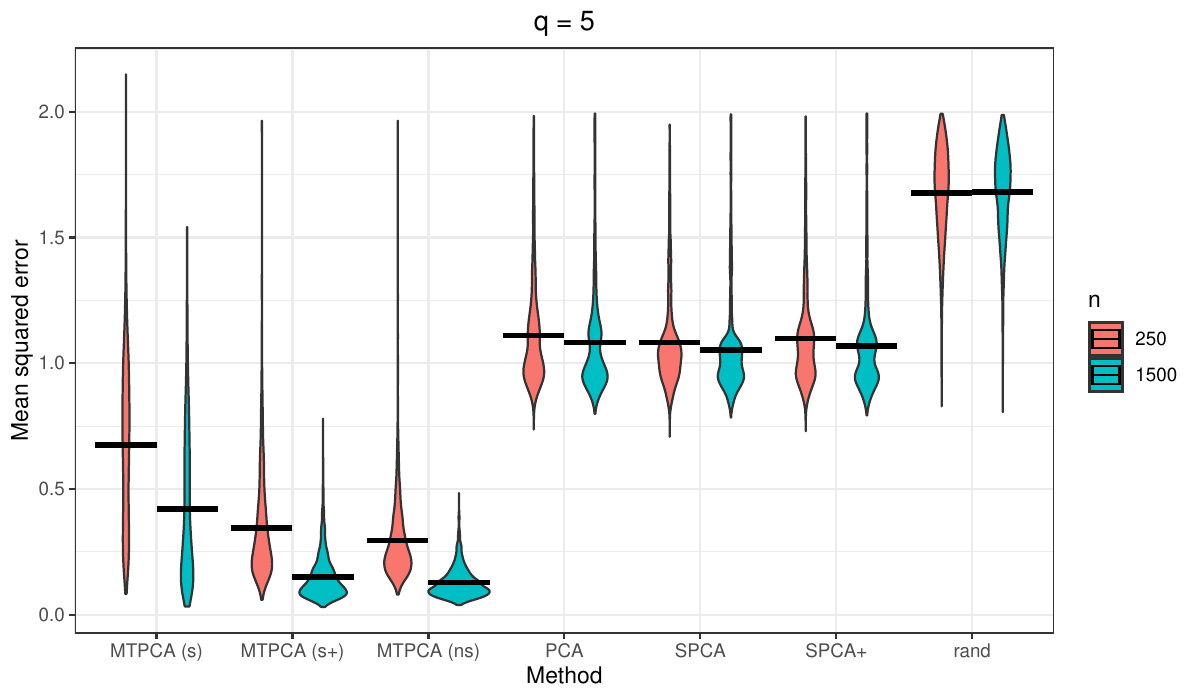}
    \caption{Performances of different methods when the simulated dataset consists of $n=250$ or $n=1500$ observations with $q=5$ out of $p=8$ non-zero coefficients. The $y$-axis gives the mean squared errors when extracting two principal components. The black horizontal lines mark the means of the empirical error distributions.}
    \label{fig:vertailu_q5}
\end{figure}

\begin{figure}
    \centering
    \includegraphics[width=0.9\linewidth]{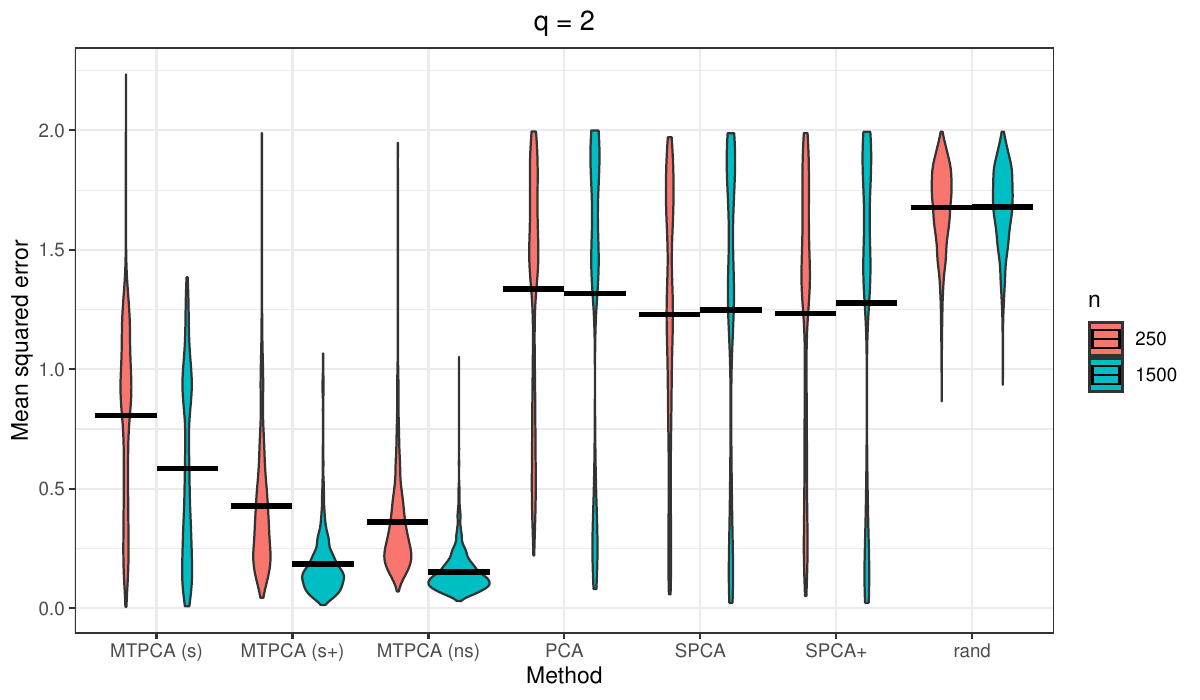}
    \caption{Performances of different methods when the simulated dataset consists of $n=250$ or $n=1500$ observations with $q=2$ out of $p=8$ coefficients non-zero. The $y$-axis gives the mean squared errors when extracting two principal components. The black horizontal lines mark the means of the empirical error distributions.}
    \label{fig:vertailu_q2}
\end{figure}

\subsection{Real data}



We apply MTPCA to the Zoo dataset \citep{zoo_111} which consists of $p = 15$ binary variables (three non-binary variables are removed). Each observational unit corresponds to one of $n = 101$ animal species and each variable describes the presence or absence of a certain physiological or ecological feature (is predator, is venomous, etc.). We use MTPCA to extract $d = 3$ latent components with $r = 5$ non-zero coefficients each. The number of components was decided based on the scree plot shown in Figure \ref{fig:Zoo_scree}, where the slope drops down after the third component. The number of non-zero coefficients $r=5$ was based on the observation that in case larger $r$ was selected, the additional non-zero loadings were typically anyway close to zero.

The resulting loadings are presented in Table \ref{tab:zoo_loadings}. We see that the first component detects the venomous animals which get a strong negative loading. These can be interpreted as outliers and this behavior closely aligns with standard PCA, where the first component tends to pick up the most deviating group of outliers due to the large variance of the associated direction. The second component separates fish (fins) and birds (feathers) from domestic land vertebrate. The third component adds to this by separating fish and birds from each other. Hence, each of the three components has a natural and meaningful interpretation.

In Figure \ref{fig:MTPCA_scores}, the animals are presented according to their component scores. Since the first component mainly shows just the separation between the venomous and non-venomous animals, we focus here on the second and third components. The scores for the different species overlap considerably which is caused by the combination of binary data and sparse loadings, leading only to a small number of possible resulting scores. Clear groups form for domestic animals, non-domestic mammals, birds, fish, aquatic mammals and other aquatic animals. In addition to this, we identify also some clearly outlying animals. 

In summary, MTPCA was able to obtain a very interpretable dimension reduction and visualization for the current data, and the analysis could be continued from here with, e.g., clustering. The main conceptual advantage of MTPCA over PCA is that regular PCA works on the level of the observed binary variables, but MTPCA models the latent features. This is especially relevant for variables such as domestic or catsize, which can be seen coming from setting a cut-off point to a more complex latent feature.


\begin{figure}
    \centering
    \includegraphics[width=0.7\linewidth]{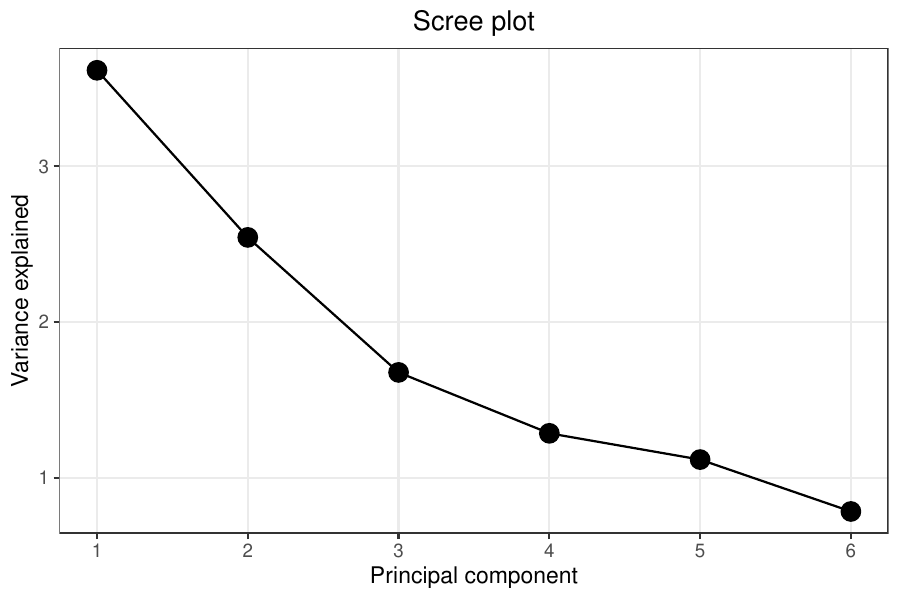}
    \caption{Scree plot for the number of components in the Zoo example}
    \label{fig:Zoo_scree}
\end{figure}

\begin{table}
\centering
\begin{tabular}{r|ccc}
 variable & PC1 & PC2 & PC3 \\ 
  \hline
  hair & $0$ & $0$ & $0$ \\ 
  feathers & $0.14$ & $0.45$ & $0.58$ \\ 
  eggs & $0$ & $0$ & $0$ \\ 
  milk & $0$ & $0$ & $0$ \\ 
  airborne & $-0.099$ & $0$ & $0.059$ \\ 
  aquatic & $0$ & $0$ & $0$ \\ 
  predator & $0$ & $0$ & $0$ \\ 
  toothed & $0$ & $0$ & $0$ \\ 
  backbone & $0$ & $-0.32$ & $-0.17$ \\ 
  breathes & $0.037$ & $-0.40$ & $0$ \\ 
  venomous & $-0.98$ & $0$ & $0.17$ \\ 
  fins & $-0.078$ & $0.36$ & $-0.78$ \\ 
  tail & $0$ & $0$ & $0$ \\ 
  domestic & $0$ & $-0.63$ & $0$ \\ 
  catsize & $0$ & $0$ & $0$ \\ 
\end{tabular}
\caption{Loadings of the original variables to the three extracted latent components for the Zoo data}
\label{tab:zoo_loadings}
\end{table}


\begin{figure}
    \centering
    \includegraphics[width=0.95\linewidth]{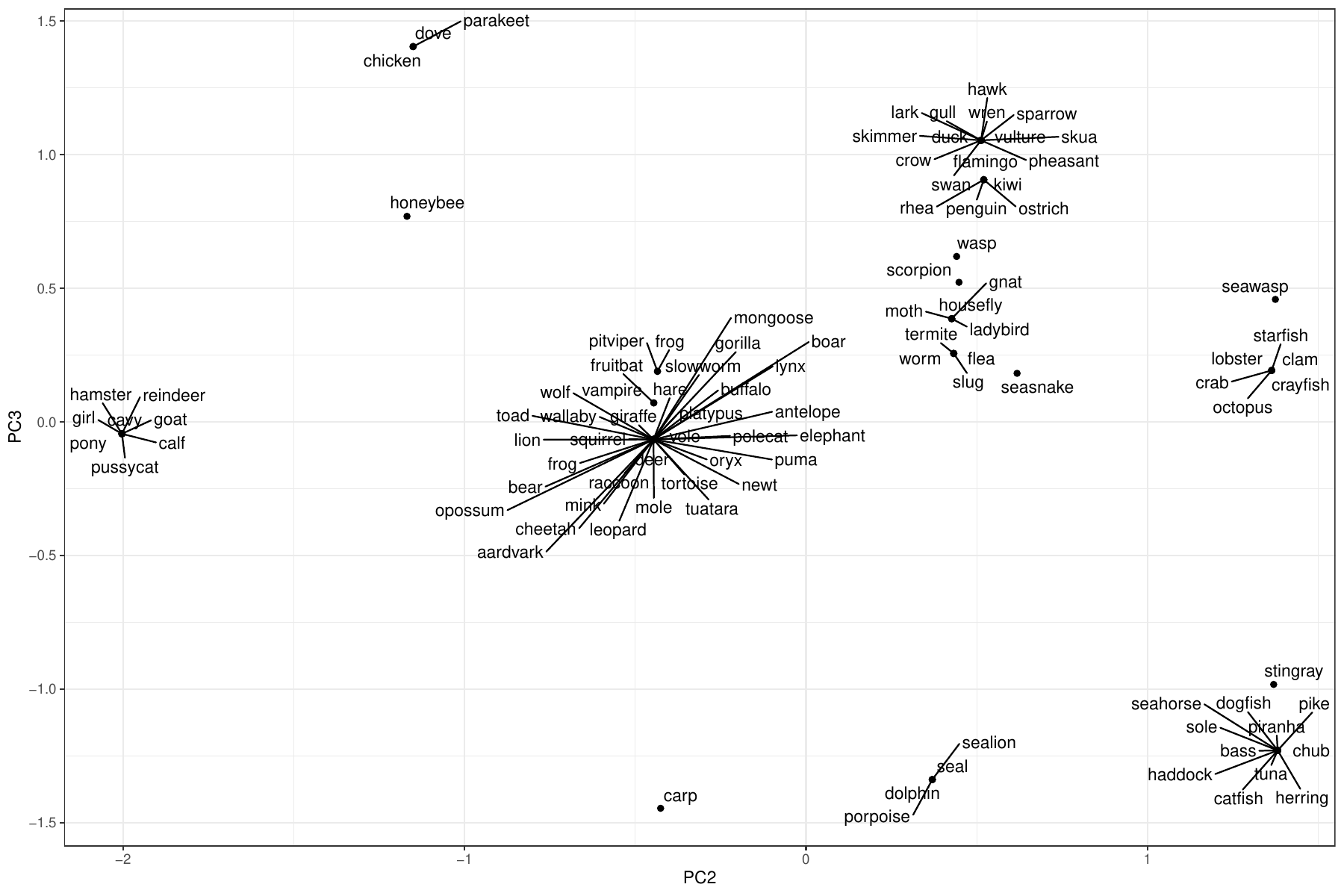}
    \caption{Scores of second and third components given by MTPCA}
    \label{fig:MTPCA_scores}
\end{figure}

\section{Discussion}


We proposed a principal component method for data consisting of different types of variables. The method is based on exponential family distributions and the method of moments. Because we give explicit formulas for the covariances of the latent combinations $u'Z$, direct asymptotics could be studied for the estimators (beyond consistency), and this is a natural future task.

Another research question concerns the consistent estimation of $d$, possibly in the manner described in Section \ref{sec:number_of_components}. Finally, additional distributions for the observed data could be considered beyond the four types we used, for example, various zero-inflated count distributions to account for a wider range of practical scenarios.


\subsection*{Author Contributions}

\textbf{Lauri Heinonen}: Conceptualization, Methodology, Software, Formal analysis, Investigation, Data Curation, Writing - Original Draft, Writing - Review \& Editing, Visualization, Project administration. \textbf{Joni Virta}: Conceptualization, Methodology, Formal analysis, Investigation, Writing - Review \& Editing, Supervision, Funding acquisition.




\subsection*{Conflicts of Interest}

The authors declare no conflicts of interest.

\bibliographystyle{chicago}
\bibliography{references}



\appendix

\section{Auxiliary results and proofs}\label{sec:proofs}
\setcounter{equation}{0}

\vspace*{12pt}


Throughout the proofs, $\varphi(\cdot)$ denotes the PDF of the standard normal distribution.

\begin{lemma}\label{lem:exp_integral}
Let $Z \sim \mathrm{N}_p(0,I_p)$, $b_1, b_2, h, v \in \mathbb{R}^p$ and $a \in \mathbb{R}$. Then
    \begin{align*}
        \mathrm{E}(e^{b_1' Z}) &= e^{\frac{1}{2} b_1' b_1}, \\
        \mathrm{E}(b_1' Ze^{b_2' Z}) &= b_1' b_2 e^{\frac{1}{2}b_2' b_2}, \\
        \mathrm{E}(Z \Phi(a + h' Z)) &= \frac{1}{\sqrt{1 + \| h \|^2}} \varphi \left( \frac{a}{\sqrt{1 + \| h \|^2}}\right) h, \\
        \mathrm{E}(e^{v'Z} \Phi(a + h' Z)) &= e^{\frac{1}{2} \| v \|^2} \Phi \left( \frac{a + h'v}{\sqrt{1 + \| h \|^2}} \right).
    \end{align*}
\end{lemma}

\begin{proof}[Proof of Lemma \ref{lem:exp_integral}]
    The first two parts can be proven by partial integration and we omit their proof. For the third part, assume that $h \neq 0$ and let $G = (g_1, \ldots, g_p)$ denote any $p \times p$ orthogonal matrix such that $g_1 = h/\| h \|$. Letting then $Y := G' Z \sim \mathrm{N}_p(0,I_p)$, we get
    \begin{align*}
        \mathrm{E}(Z \Phi(a + h' Z)) = GG'\mathrm{E}(Z \Phi(a + \| h \| g_1' Z)) = G \mathrm{E}(Y \Phi(a + \| h \| y_1)).
    \end{align*}
    By independence, only the first element of $\mathrm{E}(Y \Phi(a + \| h \| y_1))$ is non-zero and, by \cite{owen1980table}, equals
    \begin{align*}
        \frac{\| h \|}{\sqrt{1 + \| h \|^2}} \varphi \left( \frac{a}{\sqrt{1 + \| h \|^2}}\right),
    \end{align*}
    from which the claim follows for $h \neq 0$. If $h = 0$, the claim is straightforwardly checked to hold in this case as well.

    For the fourth part, we write
    \begin{align*}
        \mathrm{E}(e^{v'Z} \Phi(a + h' Z)) &= \frac{1}{(2 \pi)^{p/2}} \int e^{-\frac{1}{2} \| z \|^2} e^{v'z} \Phi(a + h'z) dz \\
        &= \frac{1}{(2 \pi)^{p/2}} e^{\frac{1}{2} \| v \|^2 } \int e^{-\frac{1}{2} \| z - v \|^2} \Phi(a + h'z) dz \\
        &= e^{\frac{1}{2} \| v \|^2 } \mathrm{E} \{ \Phi(a + h' Y) \},
    \end{align*}
    where $Y = v + Q$ and $Q \sim \mathrm{N}_p(0, I_p)$. Hence,
    \begin{align*}
        \mathrm{E}(e^{v'Z} \Phi(a + h' Z)) = e^{\frac{1}{2} \| v \|^2 } \mathrm{E}\{ \Phi(a + h' v + h' Q) \}
    \end{align*}
    and the result now follows from \cite{owen1980table}.
\end{proof}

Let us denote, in the following theorem, another variable of the same type $i$ with $X_{i*}$ (so that the symbol $'$ can be reserved for the transpose).
\begin{theorem}\label{theo:moments}
We have
    \begin{align}
    \mathrm{E}(X_{1}) &= \mu_{1} \\
    \mathrm{E}(X_{2}) &= e^{-\mu_{2} + \frac{1}{2} u_{2}' \Lambda u_{2}} =: \alpha_{2}\\
    \mathrm{E}(X_{3}) &= e^{\mu_{3} + \frac{1}{2} u_{3}' \Lambda u_{3}} =: \alpha_{3}\\
    \mathrm{E}(X_{4}) &= \Phi\left(\frac{\mu_{4}}{\sqrt{1+u_{4}' \Lambda u_{4}}}\right) =: \alpha_{4}  \\
    \mathrm{E}(X_{1}^2) &= \mu_{1}^2 + u_{1}' \Lambda u_{1} + \tau_j^2 \\
    \mathrm{E}(X_{2}^2) &= 2\alpha_{2}^2 e^{u_{2}' \Lambda u_{2}} \\
    \mathrm{E}(X_{3}^2) &=  \alpha_{3} + \alpha_{3}^2 e^{u_{3}' \Lambda u_{3}}\\
    \mathrm{E}(X_{1} X_{1*}) &= \mu_{1} \mu_{1*} + u_{1}' \Lambda u_{1*} \\
    \mathrm{E}(X_{2} X_{2*}) &= \alpha_{2}\alpha_{2*}e^{u_{2}' \Lambda u_{2*}} \\ 
    \mathrm{E}(X_{3} X_{3*}) &= \alpha_{3} \alpha_{3*} e^{u_{3}' \Lambda u_{3*}} \\
    \mathrm{E}(X_{1} X_{2}) &= \left(\mu_{1} - u_{1}' \Lambda u_{2}\right) \alpha_{2}\\
    \mathrm{E}(X_{1} X_{3}) &= \left(\mu_{1} + u_{1}' \Lambda u_{3}\right) \alpha_{3} \\
    \mathrm{E}(X_{1} X_{4}) &= \mu_{1} \alpha_{4} + \frac{u_{1}' \Lambda u_{4}}{\sqrt{1+u_{4}' \Lambda u_{4}}}\varphi\left(\frac{\mu_{4}}{\sqrt{1+u_{4}' \Lambda u_{4}}}\right) \\
    \mathrm{E}(X_{2} X_{3}) &=  \alpha_{2} \alpha_{3} e^{-u_{2}' \Lambda u_{3}} \\ 
    \mathrm{E}(X_{2} X_{4}) &= \alpha_{2} \Phi \left( \frac{\mu_{4} - u_{2}' \Lambda u_{4}}{\sqrt{1+u_{4}' \Lambda u_{4}}} \right) \\
    \mathrm{E}(X_{3} X_{4}) &= \alpha_{3} \Phi \left( \frac{\mu_{4} + u_{3}' \Lambda u_{4}}{\sqrt{1+u_{4}' \Lambda u_{4}}} \right)
\end{align}
\end{theorem}

\begin{proof}[Proof of Theorem \ref{theo:moments}]
(1) $\mathrm{E}(X_{1}) = \mathrm{E}(\mathrm{E}(X_{1}|Z)) = \mathrm{E}(W_{1}) = \mu_{1} + u_{1}' \mathrm{E}(Z) = \mu_{1}$. \\
(2) We use Lemma \ref{lem:exp_integral} to obtain $\mathrm{E}(X_{2}) = \mathrm{E}(\mathrm{E}(X_{2}|Z)) = \mathrm{E}(e^{-W_{2}}) = \mathrm{E}(e^{-\mu_{2} - u_{2}' Z}) = e^{-\mu_{2}}\mathrm{E}(e^{- u_{2}' Z}) = e^{-\mu_{2} + \frac{1}{2} u_{2}' \Lambda u_{2}}$. \\
(3) Similar to (2). \\
(4) Let $Q \sim \mathrm{N}_p(0,I)$, so $Z = \Lambda^{1/2}Q$. Then $\mathrm{E}(X_{4}) = \mathrm{E}(\mathrm{E}(X_{4}|Z)) = \mathrm{E}(\Phi(W_{4})) = \int_{-\infty}^\infty \Phi\left(\mu_{4} + (u'_{4}\Lambda u_{4})^{1/2} q\right)\varphi(q)dq = \Phi\big(\frac{\mu_{4}}{\sqrt{1+u_{4}' \Lambda u_{4}}}\big)$. The last equality is an integral found in \cite{owen1980table}.\\
(5) Let us write $X_{1} = \mu_{1} + u_{1}' Z + \varepsilon$, where $\varepsilon \sim \mathrm{N}(0,\tau_j^2)$ and $\mathrm{E}(Z)=\mathrm{E}(\varepsilon)=0$ and $\varepsilon \perp Z$. Then $\mathrm{E}(X_{1}^2) = \mathrm{E}(\mathrm{E}(X_{1}^2|Z)) = \mathrm{E}((\mu_{1} + u_{1}' Z + \varepsilon)^2) = \mu_{1}^2 + u_{1}'\mathrm{E}(ZZ')u_{1} + \mathrm{E}(\varepsilon^2) + 2\mu_{1} u_{1}'\mathrm{E}(Z) + 2\mu_{1}\mathrm{E}(\varepsilon) + 2u_{1}'\mathrm{E}(\varepsilon Z) = \mu_{1}^2 + u_{1}'\Lambda u_{1} + \tau_j^2$. \\
(6) We remember that if $Y \sim \mathrm{Exp}(\lambda)$, then $\mathrm{E}(Y^2)=\frac{2}{\lambda^2}$. Now $\mathrm{E}(X_{2}^2) = \mathrm{E}(\mathrm{E}(X_{2}^2|Z)) = \mathrm{E}(2e^{-2 W_{2}}) = 2\mathrm{E}(e^{-2\mu_{2} - 2u_{2}' Z}) = 2e^{-2\mu_{2}}\mathrm{E}(e^{-2u_{2}' Z}) = 2e^{-2\mu_{2} + \frac{1}{2} \times (-2)u_{2}' \Lambda (-2)u_{2}} = 2e^{-2\mu_{2} + 2u_{2}' \Lambda u_{2}}$, which can be written as $2e^{-2(\mu_{2} -\frac{1}{2} u_{2}' \Lambda u_{2})+u_{2}' \Lambda u_{2}} = 2\alpha_{2}^{2} e^{u_{2}' \Lambda u_{2}}$. \\
(7) similarly to (6): $Y \sim \mathrm{Poi}(\lambda) \implies \mathrm{E}(Y^2)=\lambda^2+\lambda$. Now $\mathrm{E}(X_{3}^2) = \mathrm{E}(\mathrm{E}(X_{3}^2|Z)) = \mathrm{E}(e^{2 W_{3}}+e^{W_{3}}) = \mathrm{E}(e^{2 W_{3}}) + \mathrm{E}(e^{W_{3}})$. The second part is $\alpha_{3}$ and the first, based on similar reasoning as (5), is $\alpha_{3}^2 e^{u_{3}' \Lambda u_{3}}$. \\
(8) Let us write again $X_{1} = \mu_{1} + u_{1}' Z + \varepsilon$ and similarly for $X_{1*}$ where $\varepsilon \perp \varepsilon_*$. We see that $X_{1} \perp X_{1*}\ |\ Z$. Then $\mathrm{E}(X_{1}X_{1*}) = \mathrm{E}(\mathrm{E}(X_{1}X_{1*}|Z)) = \mathrm{E}(\mathrm{E}(X_{1}|Z)\mathrm{E}(X_{1*}|Z))= \mathrm{E}((\mu_{1} + u_{1}' Z)(\mu_{1*} + u_{1*}' Z)) = \mu_{1}\mu_{1*} + u_{1}'\mathrm{E}(ZZ')u_{1*} + \mu_{1} u_{1*}'\mathrm{E}(Z) + \mu_{1*} u_{1}'\mathrm{E}(Z) = \mu_{1}\mu_{1*} + u_{1}'\Lambda u_{1*}.$ \\
(9) $\mathrm{E}(X_{2}X_{2*}) = \mathrm{E}(\mathrm{E}(X_{2}X_{2*}|Z)) = \mathrm{E}(\mathrm{E}(X_{2}|Z)\mathrm{E}(X_{2*}|Z))= \mathrm{E}(e^{-W_{2}}e^{-W_{2*}}) = \mathrm{E}(e^{-\mu_{2} - \mu_{2*} - (u_{2} + u_{2*})' Z}) = e^{-\mu_{2} - \mu_{2*} + \frac{1}{2}(u_{2} + u_{2*})' \Lambda (u_{2} + u_{2*})}$, which can be written as $e^{-\mu_{2} + \frac{1}{2}u_{2}' \Lambda u_{2} - \mu_{2*} + \frac{1}{2}u_{2*}' \Lambda u_{2*} + u_{2}' \Lambda u_{2*}} = \alpha_{2}\alpha_{2*}e^{u_{2}' \Lambda u_{2*}}$. \\
(10) Similar to (9). \\
(11) Using Lemma \ref{lem:exp_integral}, we write $\mathrm{E}(X_{1}X_{2}) = \mathrm{E}(\mathrm{E}(X_{1}X_{2}|Z)) = \mathrm{E}(\mathrm{E}(X_{1}|Z)\mathrm{E}(X_{2}|Z))= \mathrm{E}(W_{1}e^{-W_{2}}) = \mathrm{E}((\mu_{1} + u_{1}' Z)e^{-\mu_{2} - u_{2}' Z}) = \mathrm{E}(e^{-\mu_{2}}(\mu_{1}e^{-u_{2}' Z} + u_{1}' Z e^{-u_{2}' Z})) = e^{-\mu_{2}}(\mu_{1}e^{\frac{1}{2}u_{2}' \Lambda u_{2}} - u_{1}' \Lambda u_{2} e^{\frac{1}{2}u_{2}' \Lambda u_{2}})$, which can be written as $\left(\mu_{1} - u_{1}' \Lambda u_{2}\right) e^{-\mu_{2} + \frac{1}{2}u_{2}' \Lambda u_{2}} = \left(\mu_{1} - u_{1}' \Lambda u_{2}\right) \alpha_{2}$.\\
(12) Similar to (11).\\
(13) Similarly to (4): $\mathrm{E}(X_{1}X_{4}) = \mathrm{E}(\mathrm{E}(X_{1}|Z)\mathrm{E}(X_{4}|Z)) = \mathrm{E}(W_{1}\Phi(W_{4})) = \mathrm{E}((\mu_{1} + u'_{1}\Lambda^{1/2}Q)\Phi(\mu_{4} + u'_{4}\Lambda^{1/2}Q)) = \mu_{1}\mathrm{E}(\Phi(\mu_{4} + u'_{4}\Lambda^{1/2}Q)) + u'_{1}\Lambda^{1/2}\mathrm{E}(Q\Phi(\mu_{4} + u'_{4}\Lambda^{1/2}Q))$, where the first term is obtainable directly from (4). For the second part, we use Lemma~\ref{lem:exp_integral} to obtain $  u'_{1}\Lambda^{1/2}\mathrm{E}(Q\Phi(\mu_{4} + u'_{4}\Lambda^{1/2}Q)) = \frac{u_{1}' \Lambda u_{4}}{\sqrt{1+u_{4}' \Lambda u_{4}}}\varphi \left( \frac{\mu_{4}}{\sqrt{1+u_{4}' \Lambda u_{4}}} \right)$.\\
(14) Similar to (11) and (12)\\
(15) $\mathrm{E}(X_{2}X_{4}) = \mathrm{E}(\mathrm{E}(X_{2}|Z)\mathrm{E}(X_{4}|Z)) = \mathrm{E}(e^{-W_{2}}\Phi(W_{4})) = \mathrm{E}(e^{-\mu_{2} - u'_{2}\Lambda^{1/2}Q}\Phi(\mu_{4} + u'_{4}\Lambda^{1/2}Q)) = e^{-\mu_{2}}\mathrm{E}(e^{- u'_{2}\Lambda^{1/2}Q}\Phi(\mu_{4} + u'_{4}\Lambda^{1/2}Q)) $ from which the claim now follows using Lemma \ref{lem:exp_integral}. \\ 
(16) Similar to (15).
\end{proof}

\begin{lemma}\label{lem:two_bernoullis}
    Let $X_1 = a_1 + b_1' Z$ and $X_2 = a_2 + b_2' Z$ where $Z \sim N_d(0, I_d)$ and $a_1, a_2 \in \mathbb{R}$ and $b_1, b_2 \in \mathbb{R}^d$ are arbitrary. Then
    \begin{align*}
        & \mathrm{E}\{ \Phi(X_1) \Phi(X_2) \} = \Phi_2\left( \frac{a_1}{\sqrt{1 + \| b_1 \|^2}} , \frac{a_2}{\sqrt{1 + \| b_2 \|^2}}; \frac{b_1'b_2}{\sqrt{(1 + \| b_1 \|^2) (1 + \| b_2 \|^2)}} \right),
    \end{align*}
    where $\Phi_2(y_1, y_2; \rho) := \mathrm{P}(Y_1 \leq y_1, Y_2 \leq y_2)$ where $(Y_1, Y_2)$ has bivariate normal distribution with zero means, unit variances and correlation equal to $\rho$.
\end{lemma}

\begin{proof}[Proof of Lemma \ref{lem:two_bernoullis}]
    Let $Y_1, Y_2$ denote two standard univariate Gaussian variables that are independent of each other and $Z$. Then, by expressing the expected value and the normal CDFs as integrals and writing $\int_{-\infty}^{x_1} \varphi(y_1) d y_1 = \int_{-\infty}^\infty \mathbb{I}(y_1 \leq x_1) \varphi(y_1) d y_1$, we see that
    \begin{align*}
        \mathrm{E}\{ \Phi(X_1) \Phi(X_2) \} = \mathrm{P}(Y_1 \leq X_1, Y_2 \leq X_2).
    \end{align*}
    The joint distribution of $(Y_1 - X_1, Y_2 - X_2)$ is then bivariate Gaussian with the mean vector $(-a_1, -a_2)'$ and the covariance matrix
    \begin{align*}
        \begin{pmatrix}
            1 + \| b_1 \|^2 & b_1' b_2 \\
            b_1' b_2 & 1 + \| b_2 \|^2,
        \end{pmatrix}
    \end{align*}
    giving
    \begin{align*}
        & \mathrm{P}(Y_1 \leq X_1, Y_2 \leq X_2) \\ =& \mathrm{P}(Y_1 - X_1 \leq 0, Y_2 - X_2 \leq 0) \\
        =& \mathrm{P}\left\{ \frac{(Y_1 - X_1) + a_1}{\sqrt{1 + \| b_1 \|^2}} \leq \frac{a_1}{\sqrt{1 + \| b_1 \|^2}}, \frac{(Y_2 - X_2) + a_2}{\sqrt{1 + \| b_2 \|^2}} \leq \frac{a_2}{\sqrt{1 + \| b_2 \|^2}} \right\},
    \end{align*}
    from which the claim follows after noting that the correlation between $Y_1 - X_1$ and $Y_2 - X_2$ is
    \begin{align*}
        \frac{b_1'b_2}{\sqrt{(1 + \| b_1 \|^2) (1 + \| b_2 \|^2)}}.
    \end{align*}
\end{proof}

The next theorem now follows directly from Lemma \ref{lem:two_bernoullis}.


\begin{theorem}\label{theo:moments_2}
    For two Bernoulli-variates $X_{4}, X_{4*}$, we have under Identifiability constraint \ref{id:bernoulli} that
    \begin{align*}
        \mathrm{E}(X_{4} X_{4*}) = \Phi_2\left( s_{4} \sqrt{ \sigma_{4}^2 } , s_{4*} \sqrt{ \sigma_{4*}^2 }, \frac{\sigma_{44*}}{\sqrt{(1 + \sigma_{4}^2) (1 + \sigma_{4*}^2)}} \right).
    \end{align*}
    
\end{theorem}


\begin{proof}[Proof of Theorem \ref{theo:z_prediction}]
    From the Bayes rule we get $f_{Z \mid X}(z \mid x) = f_{X \mid Z}(x \mid z) f_{Z}(z)/ f_{X}(x)$, where the denominator does not depend on $z$ and the term $f_{X \mid Z}(x \mid z)$ factorizes, by the variables' conditional independence, into a product of the conditional densities of the individual $X$'s. Each Gaussian part thus produces the likelihood
    \begin{align*}
       \exp \left\{ -\frac{1}{2 \tau_j^2} (x_{1j} - \mu_{1j} - u_{1j}'z)^2 \right\}, 
    \end{align*}
    whereas each exponential observation produces
    \begin{align*}
        \exp(\mu_{2j} + u_{2j}' z) \exp\{ - x_{2j} \exp(\mu_{2j} + u_{2j}' z) \}.
    \end{align*}
    Similarly, the Poisson-observations contribute
    \begin{align*}
        \exp\{ x_{3j} (\mu_{3j} + u_{3j}' z) \} \exp\{ -\exp(\mu_{3j} + u_{3j}' z) \}
    \end{align*}
    and the Bernoulli-observations
    \begin{align*}
        \{ \Phi(\mu_{4j} + u_{4j}' z)  \}^{x_{4j}} \{ 1 - \Phi(\mu_{4j} + u_{4j}' z)  \}^{1 - x_{4j}}.
    \end{align*}
    Finally, the density $f_{Z}$ produces the part $\exp((-1/2) z' \Lambda^{-1} z)$, giving the desired form.

    For the concavity, we first note that $z \mapsto (-1/2) z' \Lambda^{-1} z)$ is strictly concave since $\Lambda$ is positive definite and, hence, it suffices to show that all other the summands in the conditional log-density are concave. This follows by noting that each of them is either a linear function or a composition of some function $g$ with a linear map of the form ``$\mu + u' z$''. And each of these functions $g$, namely $x \mapsto -(c - x)^2$, $x \mapsto -\exp(x)$, $x \mapsto \log \Phi(x)$ and $x \mapsto \log \{ 1 - \Phi(x) \}$, are easily checked to be concave. For the last two, this follows from the inequality $x \Phi(x) + \varphi(x) > 0$, which is seen to hold as follows: $x \mapsto g(x) := x \Phi(x) + \varphi(x)$ can be checked to have everywhere positive derivative. Moreover, $g(0) > 0$, meaning that $g(x) > 0$ holds once we show that $x \Phi(x) \rightarrow 0$ when $x \rightarrow -\infty$. But this follows instantly as the left tail of the normal CDF vanishes at superlinear rate. The strict concavity thus follows.

    The existence and the uniqueness of the maximizer follows from strict concavity once we show that $z \mapsto \ell(z \mid x)$ is \textit{coercive} in the sense that if $\| z \| \rightarrow \infty$ then $\ell(z \mid x) \rightarrow -\infty$. To see this, we write
    \begin{align*}
        \ell(z \mid x) \leq& C - \frac{1}{2} z' \Lambda^{-1} z + \sum_{j = 1}^{p_2}( \mu_{2j} + u_{2j}'z) + \sum_{j = 1}^{p_3} x_{3j}(\mu_{3j} + u_{3j}'z) \\
        +& \sum_{j = 1}^{p_4} \{ x_{4j} \log \Phi(\mu_{4j} + u_{4j}' z) + (1 - x_{4j}) \log [1 - \Phi(\mu_{4j} + u_{4j}' z) ] \} \\
        \leq& C - \frac{1}{2} z' \Lambda^{-1} z + \sum_{j = 1}^{p_2}( \mu_{2j} + u_{2j}'z) + \sum_{j = 1}^{p_3} x_{3j}(\mu_{3j} + u_{3j}'z),
    \end{align*}
    where the quadratic term $-z' \Lambda^{-1} z/2$ dominates when $\| z \| \rightarrow \infty$ and the positive-definiteness of $\Lambda$ ensures that the expression goes to $-\infty$, concluding the proof. 
\end{proof}


\end{document}